\documentclass[11pt,letterpaper]{article}
\usepackage[ansinew]{inputenc}
\usepackage[english]{babel}
\usepackage{latexsym,amsbsy,amssymb,amsmath,amsfonts}
\usepackage{amsthm}
\usepackage{epsfig}
\usepackage{euscript}
\usepackage{fullpage}
\usepackage{multirow,multicol,booktabs}
\usepackage[linesnumbered,ruled]{algorithm2e}
\usepackage{nicefrac}
\usepackage{tikz}
\usetikzlibrary{calc}
\usepackage[bf]{caption}
\usepackage{graphicx}
\usepackage{enumitem}
\usepackage{hyperref}
\hypersetup{pdfauthor={Christian Glaßer, Christian Reitwießner, Maximilian Witek},
            pdftitle={Balanced Combinations of Solutions in Multi-Objective Optimization}}

\newcommand{\Sum}{\sum\limits}

\newtheorem{dummytheorem}{Dummy-Theorem}[section]
\newtheorem{definition}[dummytheorem]{Definition}
\newtheorem{lemma}[dummytheorem]{Lemma}
\newtheorem{theorem}[dummytheorem]{Theorem}
\newtheorem{proposition}[dummytheorem]{Proposition}
\newtheorem{corollary}[dummytheorem]{Corollary}

\newtheorem{claim}[dummytheorem]{Claim}

\newcommand{\uint}{\mathbb{N}}
\newcommand{\sint}{\mathbb{Z}}
\newcommand{\real}{\mathbb{R}}

\newcommand{\card}[1]{\|#1\|}

\newcommand{\oli}[1]{\overline{#1}}
\newcommand{\half}[1][1]{\nicefrac{#1}{2}}

\newcommand{\eps}{\varepsilon}
\renewcommand{\card}[1]{\##1}

\newcommand{\pn}[1]{\textnormal{#1}}

\newcommand{\algosettings}{
\SetAlgoNoLine
\SetKwSty{texttt}
\SetFuncSty{texttt}
\SetArgSty{texttt}
\SetDataSty{texttt}
\SetAlFnt{\texttt}
\SetNlSty{texttt}{\normalsize}{}
\NoCaptionOfAlgo
\SetKwInOut{Input}{\textrm{Input}}
\SetKwInOut{Output}{\textrm{Output}}
}

\newcommand{\ATSP}{\pn{MaxATSP}}
\newcommand{\bfATSP}{\pn{\bf MaxATSP}}
\newcommand{\ATSPApprox}{\texttt{$k$-MaxATSP-Approx}$_\texttt{R}$}

\newcommand{\MM}{\pn{MM}}
\newcommand{\bfMM}{\pn{\bf MM}}
\newcommand{\kMMApprox}{\texttt{$k$-MM-Approx}$_\texttt{R}$}

\newcommand{\MaxSAT}{\mbox{\rm MaxSAT}}
\newcommand{\bfMaxSAT}{\mbox{\textbf{MaxSAT}}}
\newcommand{\kMaxSAT}[1][k]{\mbox{\rm $#1$-{MaxSAT}}}
\newcommand{\kMaxSATbf}[1][k]{\mbox{\bf $\boldsymbol{#1}$-{MaxSAT}}}
\newcommand{\kWSATApprox}{\mbox{\tt $2k$-MaxSAT-Approx}}

\begin{document}
\selectlanguage{english}

\title{Balanced Combinations of Solutions\\ in Multi-Objective Optimization}

\author{
Christian Glaßer \hspace*{6mm} 
Christian Reitwießner \hspace*{6mm} 
Maximilian Witek\\[5.7mm] 
{Julius-Maximilians-Universität Würzburg, Germany} \\
{\tt \{glasser,reitwiessner,witek\}@informatik.uni-wuerzburg.de}}

\date{}

\maketitle

\begin{abstract}
    \noindent For every list of integers $x_1, \ldots, x_m$ there is some $j$ such that $x_1 + \cdots + x_j - x_{j+1} - \cdots - x_m \approx 0$.
    So the list can be nearly balanced and for this we only need one alternation between addition and sub\-traction.
    But what if the $x_i$ are $k$-dimensional integer {\em vectors}?
    Using results from topological degree theory we show that balancing is still possible, now with $k$ alternations.

    This result is useful in multi-objective optimization,
    as it allows a polynomial-time computable balance of two alternatives with conflicting costs.
    The application to two multi-objective optimization problems yields the following results:
    \begin{itemize}
        \item 
        A randomized $\nicefrac{1}{2}$-approximation
        for multi-objective maximum asymmetric traveling salesman,
        which improves and simplifies the best known approximation for this problem.
        \item A deterministic $\nicefrac{1}{2}$-approximation for 
        multi-objective maximum weighted satisfiability.
    \end{itemize}

\end{abstract}


\section{Introduction} \label{sec_intro}

\paragraph{Balancing Sums of Vectors.}
Suppose we are given a sequence of goods $g_1, \ldots, g_m$,
each of which has a value, a weight, and a size.
Is it possible to distribute the goods on two trucks
such that the loads are nearly the same with respect to value, weight, and size?
We show that this is always possible by a very easy partition:
For suitable indices $i,j,k,l$,
assign $g_i, g_{i+1}, \ldots, g_{j}, \; g_{k}, g_{k+1}, \ldots, g_l$ to the first truck
and the remaining goods to the second one.
In general, if the goods have $2k$ criteria (value, weight, size, \ldots),
then there exist $k$ intervals of goods such that
the goods inside and the goods outside of the intervals
are nearly equivalent with respect to {\em all} criteria.

More formally, let $x_1, \ldots, x_m \in \uint^{2k}$ be {\em vectors of natural numbers} that represent the criteria of each good,
and let $z \in \uint^{2k}$ be an upper bound for these vectors (i.e., $x_i \le z$ for all $i$).
Lemma~\ref{lem:real:balancing} provides intervals $I_1, \ldots, I_k \subseteq \uint$
such that for $I = \bigcup_{i=1}^k I_i$,
\[
    -4kz \;\; \le \;\; \sum_{i \in I} x_i - \sum_{i \notin I} x_i \;\; \le \;\;
    4kz,
\]
where the $\le$ hold with respect to each component.
The same is true if $x_1,\ldots,x_m \in \sint^{2k}$ are {\em vectors of integers},
where $-z \le x_i \le z$ for all $i$ (Corollary~\ref{coro:combinatorial:integer}).
The proofs of these balancing results are based on the Odd Mapping Theorem,
a result from topological degree theory, which we apply in a discrete setting.
The discretization is responsible for the term $4kz$,
which is caused by a rounding error that unavoidably occurs at the boundaries of the intervals $I_1, \ldots, I_k$.

The simplicity of the desired partition (i.e., a union of $k$ intervals)
is important for the application of our balancing results.
Algorithmically, it means that for fixed dimension $2k$,
the right choice for the intervals $I_1, \ldots, I_k$
can be found by exhaustive search in time polynomial in $m$.

\paragraph{Multi-Objective Optimization.}
Many real-life optimization problems have multiple objectives that cannot be easily combined
into a single value.
Thus, one is interested in solutions that are good with respect to all objectives at the same time.
For conflicting objectives
we cannot hope for a single optimal solution, but
there will be trade-offs. The {\em Pareto set} captures the notion of optimality in this setting.
It consists of all solutions that are optimal in the sense that there is no
solution that is at least as good in all objectives and better in at least one
objective.
So the Pareto set contains all optimal decisions for a given situation.
For a general introduction to multi-objective optimization we refer to the survey by Ehrgott and Gandibleux \cite{eg00}
and the textbook by Ehrgott \cite{ehr05}.

For many problems, the Pareto set has exponential size and hence
cannot be computed in polynomial time.
Regarding the approximability of Pareto sets,
Papadimitriou and Yannakakis \cite{PY00} show that
every Pareto set has a $(1-\eps)$-approximation of size polynomial in the size of the instance and $\nicefrac{1}{\eps}$
(for the formal definition of approximation see section~\ref{sec:multi_defs}).
Hence, even though a Pareto set might be an exponentially large object,
there always exists a polynomial-size approximation.
This clears the way for a general investigation of the approximability
of Pareto sets of multi-objective optimization problems.

In general, inapproximability and hardness results
directly translate from single-objective optimization problems to their
multi-objective variants. On the other hand, existing approximation algorithms for
single-objective problems can not always be used for multi-objective approximation. 
Using our balancing results we demonstrate a translation of single-objective approximation ideas
to the multi-objective case:
We obtain
a randomized $\nicefrac{1}{2}$-approximation for multi-objective maximum asymmetric TSP and 
a deterministic $\nicefrac{1}{2}$-approximation for multi-objective maximum weighted satisfiability.

\paragraph{Traveling Salesman Problem.}
The (single-objective) maximum asymmetric traveling salesman problem (\bfATSP,
for short) is the optimization
problem where on input of a complete directed graph with edge weights from $\uint$
the goal is to find a Hamiltonian cycle of maximum weight.
Engebretsen and Karpinski \cite{EK01} show that \ATSP\ cannot be
$(319/320+\eps)$-approximated (unless P$=$NP).
In 1979, Fisher, Nemhauser and Wolsey \cite{FNW79} gave a $\nicefrac{1}{2}$-approximation
algorithm for \ATSP\ (remove the lightest edge
from each cycle of a maximum cycle cover and connect the remaining parts to a Hamiltonian cycle).
Since then, many improvements were achieved and the currently best known
approximation ratio of $\nicefrac{2}{3}$ for \ATSP\ is given by Kaplan et al.\
\cite{KLS+05}.

The $k$-objective variant {\boldmath$k${\bf-}}\bfATSP\ is defined analogously with edge weights
from $\uint^k$.  The hardness results for \ATSP\ directly translate
to its multi-objective variant (just set all but one component of the edge
weights to a constant), but
algorithms have to be newly designed. Bläser et al.\ \cite{BMP08} show that
$k$-\ATSP\ is randomized $(\frac{1}{k+1}-\eps)$-approximable.
This was improved by Manthey \cite{man09} to a randomized
$(\frac{1}{2}-\eps)$-approximation for all (fixed) numbers of criteria. Both algorithms
extend the cycle cover idea to multiple objectives.
With a surprisingly simple algorithm we improve the approximation ratio to $\nicefrac{1}{2}$.

\paragraph{Satisfiability.}
Given a formula in conjunctive normal form
and a non-negative weight in $\mathbb{N}$ for each clause,
the maximum weighted satisfiability problem (\bfMaxSAT, for short)
aims to find a truth assignment such that the sum of the weights of all satisfied clauses is maximal.
The first approximation algorithm for \MaxSAT\ is due to Johnson \cite{Joh74}.
He proved an approximation ratio of $\nicefrac{(2^r-1)}{(2^r)}$ 
for formulas where each clause has at least $r$ literals.
His work showed that the general \MaxSAT\ problem is $\nicefrac{1}{2}$-approximable.
Yannakakis \cite{Yan94} improved the approximation ratio of \MaxSAT\ to $\nicefrac{3}{4}$,
and Goemans and Williamson \cite{GW94} subsequently gave a simpler algorithm 
with essentially the same approximation ratio, and later
\cite{GW95} improved the approximation ratio to $0.758$.
Further improvements followed, and the currently best known approximation ratio of $0.7846$
is due to Asano and Williamson \cite{AW02}.
Regarding lower bounds, Papadimitriou and Yannakakis \cite{PY91} show that \MaxSAT\ is APX-complete.
Furthermore, by H{\aa}stad \cite{Has97}, \MaxSAT\ cannot be approximated better
than $\nicefrac{7}{8}$, unless P$=$NP.

Only little is known about the multi-objective maximum
weighted satisfiability problem (\kMaxSATbf, for short),
where each clause has a non-negative weight in $\mathbb{N}^k$ for some fixed $k
\geq 1$ and where we wish to maximize the weight of the satisfied clauses.
Santana et.\ al.\ \cite{SBLL09} apply genetic algorithms to a version of the problem that is
equivalent to \kMaxSAT\ with polynomially bounded weights.
To our knowledge, the approximability of \kMaxSAT\ has not been investigated so
far.

Using our balancing results, we can transfer a simple idea from single-objective
optimization to the multi-objective world:
For any truth assignment, the assignment itself or its complementary
assignment satisfies at least one half of all clauses.
We obtain a (deterministic) $\nicefrac{1}{2}$-approximation for $\kMaxSAT$,
independent of $k$.


\section{Balancing Results}
\label{combinatorics_section}

\subsection{Preliminaries}

Let $a,b \in \real$.
We call a function $f \colon [a,b] \to \real$ \emph{integrable}, if it is
Lebesgue-integrable on $[a,b]$. This is especially the case for bounded
functions $f$ with
only finitely many points of discontinuity.
A function $g \colon [a,b] \to \real^n$ is \emph{componentwise integrable},
if all projections $g_i$ are integrable and in this case
we write $\int_a^b g(x) \, dx$ as abbreviation for the tuple
$(\int_a^b g_1(x) \, dx, \ldots, \int_a^b g_n(x) \, dx)$.
For $x=(x_1,\dots,x_n)$, $y=(y_1,\dots,y_n) \in \real^{n}$ we write $x \le y$ 
if $x_i \le y_i$ for all $i\in\{1,2,\dots,n\}$.
For a set $A \subseteq \real^n$, $\oli{A}$ denotes the (topological) closure of $A$,
and $\partial A$ denotes the boundary of $A$.
The set $A\subseteq \real^n$ is \emph{symmetric} if
$x \in A \iff -x \in A$ for all $x \in \real^n$.

For bounded, open sets $D \subseteq \real^n$, continuous functions $\varphi
\colon D \to \real^n$ and points $p \in \real^n \setminus \varphi(\partial D)$
the integer $d(\varphi,D,p)$ is called the \emph{Brouwer degree} of $\varphi$
and $D$ at the point $p$.  We will not define it here, but we note that it
captures how often $p$ is ``covered'' by $\varphi(D)$, counting
``inverse'' covers negatively, and that it generalizes
the winding number in complex analysis.

\subsection{Analytical Version}

We apply the following theorems from topological degree theory to get the analytical version of our balancing results.
\begin{theorem}[\protect{\cite[Theorem 2.1.1]{ll78-degree-theory}}]\label{thm:oddmappingtwo}
If $D\subseteq\real^n$ is bounded and open,
$\varphi\colon \overline{D} \to \real^n$ is continuous,
$p \notin \varphi(\partial D)$ and
$d(\varphi,D,p)\neq 0$, then $p \in \varphi(D)$.
\end{theorem}

\begin{theorem}[Odd Mapping Theorem, \protect{\cite[Theorem
3.2.6]{ll78-degree-theory}}]\label{thm:oddmapping}
Let $D$ be a bounded, open, symmetric subset of $\real^n$ containing the
origin. If $\varphi\colon \overline{D}\to\real^n$ is continuous,
$0 \notin \varphi(\partial D)$, and for all $x \in \partial D$
it holds that
$
\frac{\varphi(x)}{|\varphi(x)|} \neq
\frac{\varphi(-x)}{|\varphi(-x)|},
$
then $d(\varphi,D,0)$ is an odd number (and in particular not zero).
\end{theorem}

\begin{corollary}\label{coro:oddmapping}
Let $D$ be a bounded, open, symmetric subset of $\real^n$ containing the
origin. If $\varphi\colon \overline{D} \to \real^n$ is continuous
and for all $x \in \partial D$ it holds that $\varphi(-x)=-\varphi(x)$, then $0
\in \varphi(\overline{D})$.
\end{corollary}
\begin{proof}
Assume that $0 \notin \varphi(\overline{D})$. From $\varphi(-x)=-\varphi(x)$
for $x \in \partial D$ it follows that the inequality condition
of Theorem~\ref{thm:oddmapping} is fulfilled (note that $0 \notin \varphi(\partial
D)$) and thus
$d(\varphi,D,0) \neq 0$ and by Theorem~\ref{thm:oddmappingtwo}, $0 \in \varphi(\overline{D})$.
This is a contradiction.
\end{proof}

\begin{lemma}\label{lem:single_annulator}
Let $n \ge 1$, $a,b \in \real$, and
$h \colon [a,b] \to \real^{2n}$ be componentwise integrable.
There exist $n$ closed intervals $I_1, \ldots, I_n \subseteq [a,b]$
such that for $I=I_1\cup \dots \cup I_n$,
\[
    \int\limits_I h(x) \,dx = \int\limits_{[a,b]\setminus I} h(x)\,dx.
\]
\end{lemma}
\begin{proof}
Observe that it suffices to show this for $[a,b]=[0,1]$.
Define $T = \{(t_1,t_2,\dots,t_{2n})\in \real^{2n} \mid \sum_{i=1}^{2n} |t_i| \le 1\}$
and for every $t=(t_1,\dots,t_{2n}) \in T$, let
\[
I_t = \bigcup_{\substack{1 \le k \le 2n\\t_k > 0}} \left[\sum_{i=1}^{k-1}|t_i|,
\sum_{i=1}^{k}|t_i|\right]
\]
and
\[
f \colon T \to \real^{2n},\quad
f(t) = \int\limits_{I_t}  h(x)\,dx - \int\limits_{[0,1]\setminus I_t} h(x)\,dx.
\]

\begin{figure}
\centering
\begin{tikzpicture}
\draw[|->] (0,0) to (9,0);
\node(positive) at (0,.6){};
\node(zero) at (0,0){};
\node(negative) at (0,-.6){};
\node[left of=zero]{$I_t$};
\node[left of=positive]{$t_i > 0$};
\node[left of=negative]{$t_i \leq 0$};
\draw[|-|] (0,.6) to node [above] {$|t_1|$} (1.0,.6);
\draw[|-|] (1.0,.6) to node [above] {$|t_2|$} (1.8,.6);
\draw[|-|] (1.8,-.6) to node [below] {$|t_3|$} (3.1,-.6);
\draw[|-|] (3.1,.6) to node [above] {$|t_4|$} (4.2,.6);
\draw[|-|] (4.2,-.6) to node [below] {$|t_5|$} (4.9,-.6);
\draw[|-|] (4.9,-.6) to node [below] {$|t_6|$} (6.1,-.6);
\draw[|-|] (6.1,-.6) to node [below] {$|t_7|$} (7.1,-.6);
\draw[|-|] (7.1,.6) to node [above] {$|t_8|$} (7.9,.6);
\draw[|-|,very thick] (0,0) to (1.8,0);
\draw[|-|,very thick] (3.1,0) to (4.2,0);
\draw[|-|,very thick] (7.1,0) to (7.9,0);
\end{tikzpicture}
\caption{Illustration of the set $I_t$ for some value of $t=(t_1,\dots,t_8)$,
where $t_1$, $t_2$, $t_4$ and $t_8$ are positive and $t_3$, $t_5$, $t_6$ and
$t_7$ are negative.}
\end{figure}
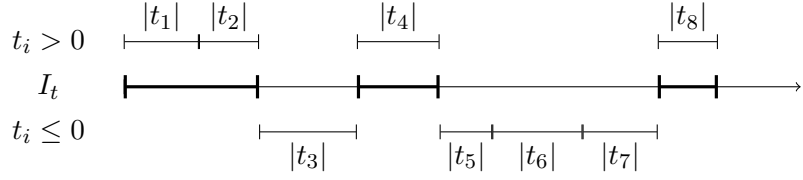

By the formal definition, $I_t$ is a union of (at most) $2n$ closed intervals.
However, it can always be written as a union of at most $n$ closed intervals,
by merging adjacent intervals.

We now want to show that $0 \in f(T)$ by applying
Corollary~\ref{coro:oddmapping} to $\varphi = f$ and
$D$ being the interior of $T$. $D$ is obviously a bounded, open, and symmetric
subset of $\real^{2n}$ containing the origin. The function $f$ is continuous because of the
fundamental theorem of calculus for the Lebesgue integral and the fact that
the endpoints of the intervals in $I_t$ depend continuously on $t$.
Furthermore, for any $t \in \partial D$ we get that there are only finitely many
points in $[0,1]$ which are not in exactly one of the sets $I_{-t}$ and
$I_t$
and thus $f(-t) = -f(t)$ since these finitely many points have no influence on the values
of the integrals. Since all preconditions of the corollary are fulfilled, we get
$0\in f(T)$ and thus there exists some
$t \in T$ such that
\[
\int\limits_{I_t} h(x)\,dx =
\int\limits_{[0,1]\setminus I_t} h(x)\,dx.
\]
As already noted, $I_t$ can be written as a union of at most $n$ closed intervals.
We obtain a union of exactly $n$ intervals by adding intervals $[a,a]$.
\end{proof}

\begin{lemma}\label{lem:double_bisector}
Let $n \ge 1$, $a,b \in \real$, and
$f,g \colon [a,b] \to \real^{2n}$ be componentwise integrable.
There exist $n$ closed intervals $I_1,\ldots,I_n \subseteq [a,b]$
such that for $I=I_1\cup \dots \cup I_n$,
\[
    \int\limits_{I} f(x) \,dx +
    \int\limits_{[a,b]\setminus I} g(x) \,dx = 
    \frac{1}{2}\int\limits_{[a,b]} f(x) + g(x) \,dx.
\]
\end{lemma}
\begin{proof}
Applying Lemma~\ref{lem:single_annulator} to $h(x) = f(x) - g(x)$ yields some
$I\subseteq [a,b]$ that is the union of $n$ closed intervals in $[a,b]$ such that
\begin{align*}
\int\limits_{I} h(x) \,dx &=
\int\limits_{[a,b]\setminus I} h(x) \,dx\\
\iff \int\limits_{I} f(x)-g(x) \,dx &=
\int\limits_{[a,b]\setminus I} f(x)-g(x) \,dx\\
\iff \int\limits_{I} f(x)-g(x) \,dx +
\int\limits_{[a,b]\setminus I} g(x)-f(x) \,dx &= 0\\
\stackrel{(*)}{\iff} 2\int\limits_{I} f(x) \,dx +
2\int\limits_{[a,b]\setminus I} g(x) \,dx &= 
\int\limits_{[a,b]} f(x) + g(x) \,dx\\
\iff \int\limits_{I} f(x) \,dx +
\int\limits_{[a,b]\setminus I} g(x) \,dx &= 
\frac{1}{2}\int\limits_{[a,b]} f(x) + g(x) \,dx.
\end{align*}
Note that $(*)$ is obtained by adding $\int_{[a,b]}f(x)+g(x)\,dx$ to both sides.
\end{proof}

\subsection{Discretization of the Analytical Results}

Now we discretize the analytical results which
causes a rounding error that cannot be avoided.

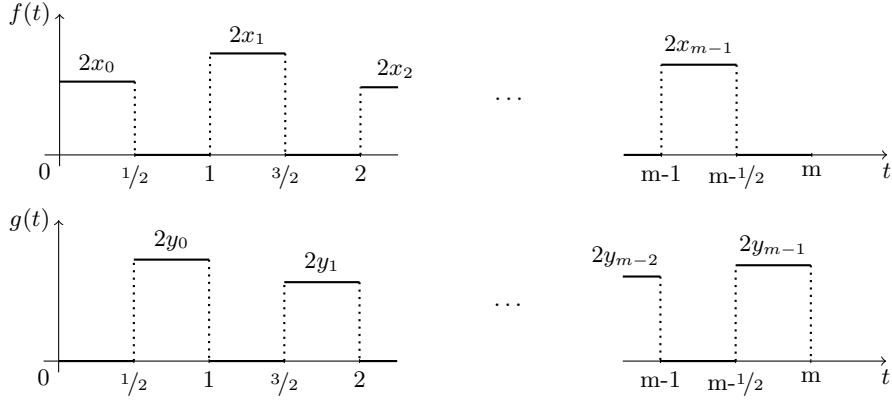
\begin{figure}
    \centering
    \footnotesize
    \begin{tikzpicture}[yscale=0.75]
    \draw (-0.2,0) node[below]{$0$} -- (4.5,0);
    \draw (6,1) node{$\dots$};
    \draw[->] (7.5,0) -- (11,0) node[below]{$t$};
    \foreach \x/\xtext in 
     {1/$\nicefrac{1}{2}$,
     2/$1$,
     3/$\nicefrac{3}{2}$,
     4/$2$,
     8/$m-1$,
     9/$m-\nicefrac{1}{2}$,
     10/$m$}
     \draw (\x cm,1pt) -- (\x cm,-1pt) node[below] {$\xtext$};
    \draw[->] (0,-0.2) -- (0,2.5) node[left]{$f(t)$};
    \draw[thick] (0,1.3) -- node[above]{$2x_0$} (1,1.3);
    \draw[thick,dotted] (1,1.3) -- (1,0);
    \draw[thick] (1,0) -- (2,0);
    \draw[thick,dotted] (2,0) -- (2,1.8);
    \draw[thick] (2,1.8) -- node[above]{$2x_1$} (3,1.8);
    \draw[thick,dotted] (3,1.8) -- (3,0);
    \draw[thick] (3,0) -- (4,0);
    \draw[thick,dotted] (4,0) -- (4,1.2);
    \draw[thick] (4,1.2) -- node[above]{$~~~~2x_2$} (4.5,1.2);
    \draw[thick] (7.5,0) -- (8,0);
    \draw[thick,dotted] (8,0) -- (8,1.6);
    \draw[thick] (8,1.6) -- node[above]{$2x_{m-1}$}(9,1.6);
    \draw[thick,dotted] (9,1.6) -- (9,0);
    \draw[thick] (9,0) -- (10,0);
    \end{tikzpicture}
    \begin{tikzpicture}[yscale=0.75]
    \draw (-0.2,0) node[below]{$0$} -- (4.5,0);
    \draw (6,1) node{$\dots$};
    \draw[->] (7.5,0) -- (11,0) node[below]{$t$};
    \foreach \x/\xtext in 
     {1/$\nicefrac{1}{2}$,
     2/$1$,
     3/$\nicefrac{3}{2}$,
     4/$2$,
     8/$m-1$,
     9/$m-\nicefrac{1}{2}$,
     10/$m$}
     \draw (\x cm,1pt) -- (\x cm,-1pt) node[below] {$\xtext$};
    \draw[->] (0,-0.2) -- (0,2.5) node[left]{$g(t)$};
    \draw[thick] (0,0) -- (1,0);
    \draw[thick,dotted] (1,0) -- (1,1.8);
    \draw[thick] (1,1.8) -- node[above]{$2y_0$} (2,1.8);
    \draw[thick,dotted] (2,1.8) -- (2,0);
    \draw[thick] (2,0) -- (3,0);
    \draw[thick,dotted] (3,0) -- (3,1.4);
    \draw[thick] (3,1.4) -- node[above]{$2y_1$} (4,1.4);
    \draw[thick,dotted] (4,1.4) -- (4,0);
    \draw[thick] (4,0) -- (4.5,0);
    \draw[thick] (7.5,1.5) -- node[above]{$2y_{m-2}~~~~$} (8,1.5);
    \draw[thick,dotted] (8,1.5) -- (8,0);
    \draw[thick] (8,0) -- (9,0);
    \draw[thick,dotted] (9,0) -- (9,1.7);
    \draw[thick] (9,1.7) -- node[above]{$2y_{m-1}$} (10,1.7);
    \draw[thick,dotted] (10,1.7) -- (10,0);
    \end{tikzpicture}
    \caption{\footnotesize Graphs of the functions $f$ and $g$ used in the proofs of the Lemmas~\ref{lem:real:balancing} and \ref{lem:combinatorial}.}
    \label{fig:vector_diagram}
\end{figure}

\begin{lemma} \label{lem:real:balancing}
    Let $n,m \ge 1$ and $x_1,\dots,x_m,y_1,\dots,y_m,z\in\uint^{2n}$ such that $x_i \le z$ and $y_i \le z$ for all $i$.
    There exist natural numbers $1 \le a_1 \le b_1 \le a_2 \le b_2 \le \cdots \le a_n \le b_n \le m$
    such that for $I = \bigcup_{i=1}^n \{a_i, a_{i}+1,\ldots, b_{i}-1\}$,
    \[
        -2nz + \frac{1}{2}\sum_{i=1}^{m} (x_i + y_i)
        \;\;\; \le \;\;\; \sum_{i \in I} x_i + \sum_{i \notin I} y_i
        \;\;\; \le \;\;\; 2nz + \frac{1}{2}\sum_{i=1}^{m} (x_i + y_i).
    \]
\end{lemma}
\begin{proof}
    For the proof it is advantageous to start the indices of $x_i$ and $y_i$ at $0$.
    We first define two functions $f$ and $g$ that distribute the values
    $x_0,\dots,x_{m-1},y_0,\dots,y_{m-1}\in\uint^{2n}$ equally over the interval
    $[0,m)$, and then we apply Lemma~\ref{lem:double_bisector}.
    Let $f,g\colon[0,m] \to \real^{2n}$ such that
    \begin{align*}
    f(t) &=
        \begin{cases}
            2x_i & \text{if $t \in [i,i+\half)$}\\
            (0,\dots,0) & \text{otherwise}
        \end{cases}\\
    \intertext{and}
    g(t) &=
        \begin{cases}
            2y_i & \text{if $t \in [i+\half,i+1)$}\\
            (0,\dots,0) & \text{otherwise.}
        \end{cases}
    \end{align*}
    Figure~\ref{fig:vector_diagram} shows the graph of $f$ and $g$.
    Note that both functions are componentwise integrable.
    Moreover, for $i\in\{0,\dots,m-1\}$,
    \begin{equation} \label{eqn_fxgya}
        \int_i^{i+1} f(t) \,dt \;=\; x_i
        \quad \mbox{and} \quad
        \int_i^{i+1} g(t) \,dt \;=\; y_i.
    \end{equation}
    By Lemma~\ref{lem:double_bisector} there exist closed intervals $I_i = [a_i,b_i] \subseteq [0,m]$ where $1 \le i \le n$
    such that for $I=\bigcup_{i=1}^{n}[a_i,b_i]$ it holds that
    \begin{equation} \label{eqn_62572}
        \int\limits_{I} f(t) \,dt +
        \int\limits_{[0,m]\setminus I} g(t) \,dt \;\;=\;\; 
        \frac{1}{2}\int\limits_{[0,m]} f(t) + g(t) \,dt.
    \end{equation}
    We may assume
    $0 \le a_1 \le b_1 \le a_2 \le b_2 \le \cdots \le a_n \le b_n \le m$.
    For $1 \le i \le n$ let
    $$a'_i := \lfloor a_i + \nicefrac{1}{2} \rfloor \quad \mbox{and} \quad b'_i := \lfloor b_i + \nicefrac{1}{2} \rfloor.$$
    Note that the $a'_i, b'_i$ are natural numbers such that
    $0 \le a'_1 \le b'_1 \le a'_2 \le b'_2 \le \cdots \le a'_n \le b'_n \le m$.
    By the definition of $f$ and $g$, for $1 \le i \le n$ it holds that
    \begin{equation*}
        \left| \int_{a_i}^{a'_i} f(t) \,dt \right| + \left| \int_{a_i}^{a'_i} g(t) \,dt \right| \le z \quad \mbox{and} \quad
        \left| \int_{b_i}^{b'_i} f(t) \,dt \right| + \left| \int_{b_i}^{b'_i} g(t) \,dt \right| \le z,
    \end{equation*}
    where $|(v_1, \ldots, v_{2n})| := (|v_1|, \ldots, |v_{2n}|)$ for
    $v_1, \ldots, v_{2n} \in \real$.
    So if some $a_i$ (resp., $b_i)$ is replaced by $a'_i$ (resp., $b'_i)$,
    then the left-hand side of (\ref{eqn_62572}) changes at most by $z$.
    Hence, for $I'=\bigcup_{i=1}^{n}[a'_i,b'_i]$ it holds that
    \begin{equation} \label{eqn_72348}
        -2nz + \frac{1}{2}\int\limits_{[0,m]} f(t) + g(t) \,dt
        \;\;\; \le \;\;\; \int\limits_{I'} f(t) \,dt \; + \!\!\!\!\! \int\limits_{[0,m]\setminus I'} g(t) \,dt
        \;\;\; \le \;\;\; 2nz + \frac{1}{2}\int\limits_{[0,m]} f(t) + g(t) \,dt.
    \end{equation}
    Let $I'' = \bigcup_{i=1}^n \{a'_i, a'_{i}+1,\ldots, b'_{i}-1\}$.
    From (\ref{eqn_fxgya}) and (\ref{eqn_72348}) we obtain
    \begin{equation*}
        -2nz + \frac{1}{2}\sum_{i=0}^{m-1} (x_i + y_i)
        \;\;\; \le \;\;\; \sum_{i \in I''} x_i + \sum_{i \notin I''} y_i
        \;\;\; \le \;\;\; 2nz + \frac{1}{2}\sum_{i=0}^{m-1} (x_i + y_i).
    \end{equation*}
\end{proof}

Next we state the integer variant of Lemma~\ref{lem:real:balancing}.
\begin{corollary} \label{coro:combinatorial:integer}
    Let $n,m \ge 1$, $x_1,\dots,x_m\in\sint^{2n}$, and $z \in \uint^{2n}$
    such that $-z \le x_i \le z$ for all $i$.
    There exist natural numbers $1 \le a_1 \le b_1 \le a_2 \le b_2 \le \cdots \le a_n \le b_n \le m$
    such that for $I = \bigcup_{i=1}^n \{a_i, a_{i}+1,\ldots, b_{i}-1\}$,
    \[
        -4nz \;\;\le\;\; \sum_{i \in I} x_i - \sum_{i \notin I} x_i \;\;\le\;\; 4nz.
    \]
\end{corollary}
\begin{proof}
    Let $x'_i := z+x_i$ and $y'_i := z-x_i$.
    Thus $x'_i, y'_i \in \uint^{2n}$ and $x'_i,y'_i \le 2z$.
    Lemma~\ref{lem:real:balancing} applied to $x'_i$ and $y'_i$ provides natural numbers $1 \le a_1 \le b_1 \le a_2 \le b_2 \le \cdots \le a_n \le b_n \le m$
    such that for $I = \bigcup_{i=1}^n \{a_i, a_{i}+1,\ldots, b_{i}-1\}$,
    \begin{equation*}
        -4nz + \frac{1}{2}\sum_{i=1}^{m} (x'_i + y'_i)
        \;\;\; \le \;\;\; \sum_{i \in I} x'_i + \sum_{i \notin I} y'_i
        \;\;\; \le \;\;\; 4nz + \frac{1}{2}\sum_{i=1}^{m} (x'_i + y'_i).
    \end{equation*}
    Therefore,
    \begin{equation*}
        -4nz + \frac{2mz}{2} + \frac{1}{2}\sum_{i=1}^{m} (x_i - x_i)
        \;\;\; \le \;\;\; mz + \sum_{i \in I} x_i - \sum_{i \notin I} x_i
        \;\;\; \le \;\;\; 4nz + \frac{2mz}{2} + \frac{1}{2}\sum_{i=1}^{m} (x_i - x_i).
    \end{equation*}
\end{proof}

For the applications to maximum asymmetric traveling salesman and maximum weighted satisfiability
we need the following variant of Lemma~\ref{lem:real:balancing}.
While providing only a lower bound for the balanced sum, it estimates the rounding error more precisely.
\begin{lemma} \label{lem:combinatorial}
    Let $n,m \ge 1$ and $x_1,\dots,x_m,y_1,\dots,y_m\in\uint^{2n}$.
    There exists an $n' \in \{0, \ldots, n\}$ and natural numbers $1 \le a_1 \le b_1 < a_2 \le b_2 < \cdots < a_{n'} \le b_{n'} \le m$
    such that for $I = \bigcup_{i \in \{1,\ldots,n'\}} \{a_i, a_{i}+1,\ldots, b_i\}$,
    \[
        y_{b_1} + y_{b_2} + \cdots + y_{b_{n'}} + \sum_{i \in I} x_i + \sum_{i \notin I} y_i \;\;\ge\;\; \frac{1}{2} \sum_{i=1}^{m}(x_i + y_i).
    \]
\end{lemma}
\begin{proof}
    Again let the indices of $x_i$ and $y_i$ start at $0$ and
    define the componentwise integrable functions $f$ and $g$ as in the proof of Lemma~\ref{lem:real:balancing}.
    So for $i\in\{0,\dots,m-1\}$,
    \begin{equation} \label{eqn_fxgy}
        \int_i^{i+\half} f(t) \,dt \;\;=\;\; x_i
        \quad \mbox{and} \quad
        \int_{i+\half}^{i+1} g(t) \,dt \;\;=\;\; y_i.
    \end{equation}
    By Lemma~\ref{lem:double_bisector} there exists an $n' \in \{0,\ldots,n\}$
    and closed intervals $I_i = [a_i,b_i] \subseteq [0,m]$ where $1 \le i \le n'$
    such that for $I=\bigcup_{i=1}^{n'}[a_i,b_i]$
    it holds that
    \begin{align}
        \int\limits_{I} f(t) \,dt +
        \int\limits_{[0,m]\setminus I} g(t) \,dt \;\;\ge\;\; 
        \frac{1}{2}\int\limits_{[0,m]} f(t) + g(t) \,dt.\label{eq:combinatorial_1}
    \end{align}
    Here we only need the inequality even though Lemma~\ref{lem:double_bisector}
    states an equality.
    We may assume
    \begin{equation} \label{eqn_order}
        0 \le a_1 \le b_1 \le a_2 \le b_2 \le \cdots \le a_{n'} \le b_{n'} \le m.
    \end{equation}

    By the definition of $f$ and $g$,
    the following holds for every $i\in\{0,\dots,m-1\}$:
    \begin{align*}
        t \in [i,i+\half) &\implies g(t) = (0,\dots,0)\\
        t \in [i+\half,i+1) &\implies f(t) = (0,\dots,0)
    \end{align*}

    \begin{claim} \label{claim_a}
        We may assume that
        $\{a_j,b_j\} \not\subseteq [i+\half,i+1]$ and
        $\{b_j,a_{j+1}\} \not\subseteq [i,i+\half]$
        for all $j \in \{1,\ldots,n'\}$ and $i\in\{0,\dots,m-1\}$.
    \end{claim}
    \begin{proof}
        If $a_j,b_j \in [i+\half,i+1]$,
        then $f$ is $0$ on $[a_j,b_j)$ and hence $\int_{a_j}^{b_j} f(t) \,dt = 0$.
        Thus the left-hand side of \eqref{eq:combinatorial_1}
        does not decrease if we remove the interval $[a_j,b_j]$ from $I$.
        Similarly, if $b_j,a_{j+1} \in [i,i+\half]$,
        then $g$ is $0$ on $[b_j,a_{j+1})$ and
        hence $\int_{b_j}^{a_{j+1}} g(t) \,dt = 0$.
        Thus the left-hand side of \eqref{eq:combinatorial_1}
        does not decrease if we replace the intervals $[a_j,b_j]$ and
        $[a_{j+1},b_{j+1}]$ by the interval $[a_j,b_{j+1}]$.
        Note that after these changes (which include a decrement of $n'$), (\ref{eqn_order}) still holds.
    \end{proof}

    \begin{claim} \label{claim_b}
        We may assume that $a_1, \ldots, a_{n'} \in \uint$ and
        $b_1 + \half, \ldots, b_{n'} + \half \in \uint$.
    \end{claim}
    \begin{proof}
        Assume $a_j \in [i+\half,i+1)$.
        By Claim~\ref{claim_a}, $b_j \notin [i+\half,i+1]$ and hence $b_j > i+1$.
        Since $f$ is $0$ on $[i+\half,i+1)$,
        the left-hand side of \eqref{eq:combinatorial_1}
        does not decrease if we let $a_j := i+1$.
        After this change, (\ref{eqn_order}) still holds.

        Assume $a_j \in [i,i+\half)$.
        By Claim~\ref{claim_a},
        $b_{j-1} \notin [i,i+\half]$ and hence $b_{j-1} < i$ (for $j \ge 2$).
        Since $g$ is $0$ on $[i,i+\half)$,
        the left-hand side of \eqref{eq:combinatorial_1}
        does not decrease if we let $a_j := i$.
        After this change, (\ref{eqn_order}) still holds.

        Assume $b_j \in [i+\half,i+1)$.
        By Claim~\ref{claim_a},
        $a_j \notin [i+\half,i+1]$ and hence $a_j < i+\half$.
        Since $f$ is $0$ on $[i+\half,i+1)$,
        the left-hand side of \eqref{eq:combinatorial_1}
        does not decrease if we let $b_j := i+\half$.
        After this change, (\ref{eqn_order}) still holds.

        Assume $b_j \in [i,i+\half)$ and $i<m$.
        By Claim~\ref{claim_a},
        $a_{j+1} \notin [i,i+\half]$ and hence $a_{j+1} > i+\half$ (for $j<n'$).
        Since $g$ is $0$ on $[i,i+\half)$,
        the left-hand side of \eqref{eq:combinatorial_1}
        does not decrease if we let $b_j := i+\half$.
        After this change, (\ref{eqn_order}) still holds.

        It remains to argue for the special case $b_j = m$.
        By Claim~\ref{claim_a},
        $a_j \notin [m-\half,m]$ and hence $a_j < m-\half$.
        Since $f$ is $0$ on $[m-\half,m)$,
        the left-hand side of \eqref{eq:combinatorial_1}
        does not decrease if we let $b_j := m-\half$.
        After this change, (\ref{eqn_order}) still holds.
    \end{proof}

    If we split the integrals on the left-hand side of (\ref{eq:combinatorial_1})
    according to $I=\bigcup_{i=1}^{n'}[a_i,b_i]$,
    we obtain
    \begin{equation} \label{eqn_31415}
        \int\limits_{0}^{a_1} g(t) \,dt
        \;+\;
        \sum_{i=1}^{n'-1}
        \left(
        \int\limits_{a_i}^{b_i} f(t) \,dt +
        \int\limits_{b_i}^{a_{i+1}} g(t) \,dt
        \right)
        \;+\;
        \int\limits_{b_{n'}}^{m} g(t) \,dt
        \;\;\;\ge\;\;\; 
        \frac{1}{2}\int\limits_{[0,m]} f(t) + g(t) \,dt.
    \end{equation}

    For $i \in \{1,\ldots,n'\}$ let $c_i = b_i-\half$.
    From Claim~\ref{claim_b} and (\ref{eqn_order}) it follows that
    $$0 \le a_1 \le c_1 < a_2 \le c_2 < \cdots < a_{n'} \le c_{n'} \le m-1.$$
    Together with (\ref{eqn_fxgy}) we obtain:
    \begin{eqnarray*}
        \int\limits_{0}^{a_1} g(t) \,dt &=& y_0 + y_1 + \cdots + y_{a_1-1}\\
        \int\limits_{a_i}^{b_i} f(t) \,dt &=& x_{a_i} + x_{a_i+1} + \cdots + x_{c_i}\\
        \int\limits_{b_i}^{a_{i+1}} g(t) \,dt &=& y_{c_i} + y_{c_i+1} + \cdots + y_{a_{i+1}-1}\\
        \int\limits_{b_{n'}}^{m} g(t) \,dt &=& y_{c_{n'}} + y_{c_{n'}+1} + \cdots + y_{m-1}
    \end{eqnarray*}
    In these sums, each index $j\in\{c_1,c_2,\ldots,c_{n'}\}$ appears exactly twice,
    once as $x_{j}$ and once as $y_{j}$.
    All remaining indices $j \in \{0,\ldots,m-1\} \setminus \{c_1, c_2, \ldots, c_{n'}\}$ appear exactly once, either as $x_j$ or as $y_j$.
    Therefore, with $I' = \bigcup_{i \in \{1,\ldots,n'\}} \{a_i,a_i + 1, \ldots, c_i\}$ the left-hand side of (\ref{eqn_31415}) is equal to
    \begin{equation} \label{eqn_31416}
        y_{c_1} + y_{c_2} + \cdots + y_{c_{n'}} + \sum_{i \in I'} x_i + \sum_{i \notin I'} y_i.
    \end{equation}
    Applying (\ref{eqn_fxgy}) to the right-hand side of (\ref{eqn_31415})
    yields the desired inequality
    \[
        y_{c_1} + y_{c_2} + \cdots + y_{c_{n'}} + \sum_{i \in I'} x_i + \sum_{i \notin I'} y_i
        \;\;\;\ge\;\;\; 
        \frac{1}{2} \sum_{i=0}^{m-1}(x_i + y_i).\qedhere
    \]
\end{proof}

\begin{corollary} \label{coro:combinatorial}
    Let $n,m \ge 1$ and $x_1,\dots,x_m,y_1,\dots,y_m,z\in\uint^{2n}$
    such that $y_i \le z$ for all $i$.
    There exist $n' \le \min(n,m)$
    disjoint, nonempty intervals $I_1, \ldots, I_{n'} \subseteq \{1,\ldots,m\}$
    such that for $I = I_1 \cup \cdots \cup I_{n'}$,
    \[
        n' \cdot z + \sum_{i \in I} x_i + \sum_{i \notin I} y_i \;\;\ge\;\; \frac{1}{2} \sum_{i=1}^{m}(x_i + y_i).
    \]
\end{corollary}


\section{Applications to Multi-Objective Optimization}

\subsection{Preliminaries} \label{sec:multi_defs} Consider some multi-objective maximization
problem $\mathcal{O}$ that consists of a set of instances $\mathcal{I}$, a set of solutions
$S(x)$ for each instance $x \in \mathcal{I}$, and a function $w$ assigning a
$k$-dimensional weight $w(x, s) \in \uint^k$ to each solution $s \in S(x)$
depending also on the instance $x \in \mathcal{I}$. If the instance $x$ is clear
from the context, we also write $w(s) = w(x, s)$.  The components of $w$ are
written as $w_i$ for $i \in \{1,2,\dots,k\}$.  For weights $a = (a_1 , \dots,
a_k)$, $b = (b_1 , \dots, b_k ) \in \uint^k$ we write $a \ge b$ if $a_i \ge b_i$
for all $i \in \{1, 2, \dots, k\}$.

Let $x \in \mathcal{I}$.
The Pareto set of $x$, the set of optimal solutions, is the set $\{s \in S(x)
\mid \neg \exists s' \in S(x)\,\, (w(x, s') \ge w(x, s) 
\text{ and } w(x, s') \neq w(x, s))\}$. For solutions $s,
s' \in S(x)$ and $\alpha < 1$ we say $s$ is $\alpha$-approximated by $s'$ if
$w_i(s') \ge \alpha \cdot w_i (s)$ for all $i$. We call a set of solutions
$\alpha$-approximate Pareto set of $x$ if every solution $s \in S(x)$ (or equivalently, every
solution from the Pareto set) is $\alpha$-approximated by some $s'$ contained in
the set.

We say that some algorithm is an $\alpha$-approximation algorithm for
$\mathcal{O}$ if
it runs in polynomial time and returns an $\alpha$-approximate Pareto set of $x$
for all input instances $x \in \mathcal{I}$. We call it randomized if it is
allowed to fail with probability at most $\nicefrac{1}{2}$ over all of its
executions.
An algorithm is an FPTAS (fully polynomial-time approximation scheme) for a
given optimization problem, if on input $x$ and $0<\eps<1$ it computes a
$(1-\eps)$-approximate Pareto set of $x$ in time polynomial in
$|x|+\nicefrac{1}{\eps}$.  If the algorithm is randomized it is called FPRAS
(fully polynomial-time randomized approximation scheme).

\subsection{\texorpdfstring{$\boldsymbol{k}$}{k}-Objective Maximum Asymmetric Traveling Salesman Problem} \label{sec_atsp}

\paragraph{Definitions.}
Let $k \geq 1$.  An \emph{$\uint^k$-labeled directed graph} is a tuple
$G=(V,E,w)$, where $V$ is some finite set of vertices, $E \subseteq V \times V$
is a set of edges, and $w \colon E \to \uint^k$ is a $k$-dimensional
weight function. We denote the $i$-th component of $w$ by $w_i$ and extend $w$ to sets
of edges by taking the sum over the weights of all edges in the set. 
A set of edges $M \subseteq E$ is called \emph{matching} in $G$ 
if no two edges in $M$ share a common vertex.
A \emph{walk} in $G$ is an alternating sequence of vertices and
edges $v_0,e_1,v_1, \dots e_{m},v_m$, where $v_i \in V$, $e_j \in E$, and
$e_j=(v_{j-1},v_j)$ for all $0 \leq i \leq m$ and $1 \leq j \leq m$.  If the
sequence of vertices $v_0,v_1,\dots,v_m$ does not contain any repetitions, the
walk is called a \emph{path} and if $v_0,v_1,\dots,v_{m-1}$ does not contain any
repetitions and $v_{m} = v_0$, it is called a \emph{cycle}.  A cycle in $G$ 
is called \emph{Hamiltonian} if it visits every vertex in $G$.  For
simplicity, we will interpret paths and cycles as sets of edges and can thus
(using the above mentioned extension of $w$ to sets of edges) write $w(C)$ for
the (multidimensional) weight of a Hamiltonian cycle $C$ of $G$.

Given some $\uint^k$-labeled directed graph as input,
our goal is to find a maximum Hamiltonian cycle.
We will also use the multi-objective version
of the maximum matching problem.
These two maximization problems are defined as follows:

\begin{quote}
\begin{tabbing}
\textbf{$\boldsymbol{k}$-Objective Maximum Asymmetric Traveling Salesman
Problem}\\ 
\textbf{($\boldsymbol{k}$-\bfATSP)}\\
Instance: \= $\uint^k$-labeled directed complete graph $(V,E,w)$\\
Solution: \> Hamiltonian cycle $C$\\
Weight: \> $w(C)$
\end{tabbing}
\end{quote}

\begin{quote}
\begin{tabbing}
{\bf $\boldsymbol{k}$-Objective Maximum Matching
($\boldsymbol{k}$-\bfMM)}\\
Instance: \= $\uint^k$-labeled directed graph $(V,E,w)$\\
Solution: \> Matching $M$\\
Weight: \> $w(M)$
\end{tabbing}
\end{quote}
Papadimitriou and Yannakakis \cite{PY00} give an FPRAS for $k$-\MM,
which we will denote by \kMMApprox\ and use as a black box in our algorithm.
Since \kMMApprox\ will be called multiple times, we
assume that its success probability is amplified in a way such that the
probability that {\em all} calls to the FPRAS succeed is at least $\nicefrac{1}{2}$.

\paragraph{High-Level Explanation of the Algorithm.} We apply the balancing results to the multi-objective maximum asymmetric traveling salesman problem
and obtain a short algorithm that provides a randomized $\nicefrac{1}{2}$-approximation.
This improves and simplifies the $(\nicefrac{1}{2}-\eps)$-approximation that was given by Manthey \cite{man09}.
Essentially our algorithm contracts a small number of edges,
then computes a maximum matching, adds the contracted edges to the matching,
and extends the result in an arbitrary way to a Hamiltonian cycle.

The argument for the correctness of the algorithm is as follows:
Each Hamiltonian cycle $H$ induces two perfect matchings 
(the edges with odd and the edges with even sequence number in the cycle).
For each objective $i$, the weight of one of the matchings is at least
$\nicefrac{1}{2} \cdot w_i(H)$.
The balancing results assure the existence of a matching $M$ such that for {\em all} objectives
the inequality $w_i(M) \ge \nicefrac{1}{2} \cdot w_i(H)$ holds up to a small error.
This matching can be approximated with the known FPRAS for multi-objective maximum matching.
Moreover, by guessing and contracting a constant number of heavy edges in $H$
our algorithm can compensate the errors caused by the balancing and by the FPRAS.


\paragraph{Contraction and Expansion of Paths.}
Suppose that for a given $\uint^k$-labeled complete directed graph $G=(V,E,w)$
we wish to find some Hamiltonian cycle that contains a particular edge $e=(u,v)$.
This reduces to the problem of finding some Hamiltonian cycle in
the $\uint^k$-labeled complete directed graph $G'=(V',E',w')$
where the edge $e$ is contracted by combining the nodes $u$ and $v$ into a
single
node while retaining the ingoing edges of $u$ and the outgoing edges of $v$.
More formally, we remove $v$ and all incident edges from $G$
and set $w'(u,x)=w(v,x)$ for every $x \in V\setminus\{u,v\}$
(Figure \ref{fig:contract}).
Now suppose we find a Hamiltonian cycle $C'$ in $G'$.
Then there exists some $x$ such that $(u,x) \in C'$. Note that $w'(u,x)=w(v,x)$.
We replace the edge $(u,x)$ in $C'$ with the detour $(u,v),(v,x)$
and obtain a Hamiltonian cycle $C$ in $G$ passing through $e$.
Moreover, $C$ preserves the weights of $C'$ in the sense that
$w(C) = w'(C') + w(e)$.

\begin{figure}
\centering
\footnotesize
\begin{tikzpicture}[scale=1]
\usetikzlibrary{calc};
\draw(0,0) node(u){$u$};
\draw(3,2) node(v){$v$};
\draw(6,0) node(w){$x$};
\draw[->, dashed] (u.north east) -- node[above left]{$e$} (v.south west);
\draw[->, dashed] (v.south east) -- node[above right]{} (w.north west);
\draw[->] (u.north east) .. controls($(u.north east)!.8!(v.south west)$) 
 and ($(v.south east)!.8!(w.north west)$) .. node[below]{} (w.north west);
\end{tikzpicture}
\caption{\footnotesize Contracting the edge $e=(u,v)$
deletes all edges incident to $v$
and sets the weights of every edge $(u,x)$ to the weights of the edge $(v,x)$
for $x \in V \setminus \{u,v\}$.
Any Hamiltonian cycle passes through some edge $(u,x)$
and hence can be expanded to a Hamiltonian cycle through $e$
by replacing $(u,x)$ with the detour $(u,v),(v,x)$.}
\label{fig:contract}
\end{figure}
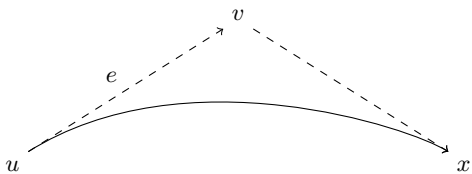

The notion of edge contractions and expansions 
can easily be extended to sets of pairwise vertex disjoint paths (Figure~\ref{fig_path_contract}),
where each path is contracted edge-by-edge starting at the last edge of the path,
and different paths can be contracted in an arbitrary order.
We make this precise with the following definitions.

\begin{figure}
\centering
\footnotesize
\begin{tikzpicture}[scale=.95,bend angle=7,pos=.3,minimum size=0mm,inner sep=1pt]
\node[circle,draw,minimum size=5mm] (b) at (0,-2) {$x$};
\node[circle,draw,minimum size=5mm] (t) at (0,2) {$y$};
\node[circle,draw,minimum size=5mm] (l) at (-2,0) {$u$};
\node[circle,draw,minimum size=5mm] (r) at (2,0) {$v$};
\draw[->] (l) to [bend left] node[auto]{$1$} (t);
\draw[->] (t) to [bend left] node[auto]{$5$} (l);
\draw[->] (t) to [bend left] node[auto]{$1$} (r);
\draw[->,very thick] (r) to [bend left] node[auto]{$3$} (t);
\draw[->] (r) to [bend left] node[auto]{$1$} (b);
\draw[->] (b) to [bend left] node[auto]{$1$} (r);
\draw[->] (b) to [bend left] node[auto]{$1$} (l);
\draw[->] (l) to [bend left] node[auto]{$1$} (b);
\draw[->,very thick] (l) to [bend left] node[auto,pos=.3]{$2$} (r);
\draw[->] (r) to [bend left] node[auto,pos=.3]{$1$} (l);
\draw[->] (t) to [bend left] node[auto,pos=.3]{$7$} (b);
\draw[->] (b) to [bend left] node[auto,pos=.3]{$1$} (t);

\node at (3.5,0) {$\xrightarrow{\mathrm{contract}_{(v,y)}}$};

\node[circle,draw,minimum size=5mm] (b2) at (6.5,-2) {$x$};
\node[circle,draw,minimum size=5mm] (l2) at (6.5-1.5,0) {$u$};
\node[circle,draw,minimum size=5mm] (r2) at (6.5+1.5,0) {$v$};
\draw[->] (r2) to [bend left] node[auto]{$7$} (b2);
\draw[->] (b2) to [bend left] node[auto]{$1$} (r2);
\draw[->] (b2) to [bend left] node[auto]{$1$} (l2);
\draw[->] (l2) to [bend left] node[auto]{$1$} (b2);
\draw[->,very thick] (l2) to [bend left] node[auto,pos=.3]{$2$} (r2);
\draw[->] (r2) to [bend left] node[auto,pos=.3]{$5$} (l2);

\node at (9.5,0) {$\xrightarrow{\mathrm{contract}_{(u,v)}}$};

\node[circle,draw,minimum size=5mm] (t3) at (11,0) {$u$};
\node[circle,draw,minimum size=5mm] (b3) at (12.5,-2) {$x$};
\draw[->] (t3) to [bend left] node[auto]{$7$} (b3);
\draw[->] (b3) to [bend left] node[auto]{$1$} (t3);

\end{tikzpicture}

\vspace{5mm}

\begin{tikzpicture}[scale=.95,bend angle=7,pos=.3,minimum size=0mm,inner sep=0pt]
\node[circle,draw,minimum size=5mm] (b) at (0,-2) {$x$};
\node[circle,draw,minimum size=5mm] (t) at (0,2) {$y$};
\node[circle,draw,minimum size=5mm] (l) at (-2,0) {$u$};
\node[circle,draw,minimum size=5mm] (r) at (2,0) {$v$};
\draw[->,very thick] (r) to [bend left] node[auto]{$3$} (t);
\draw[->] (b) to [bend left] node[auto]{$1$} (l);
\draw[->,very thick] (l) to [bend left] node[auto,pos=.3]{$2$} (r);
\draw[->] (t) to [bend left] node[auto,pos=.3]{$7$} (b);

\node at (3.5,0) {$\xleftarrow{\mathrm{expand}_{(v,y)}}$};

\node[circle,draw,minimum size=5mm] (b2) at (6.5,-2) {$x$};
\node[circle,draw,minimum size=5mm] (l2) at (6.5-1.5,0) {$u$};
\node[circle,draw,minimum size=5mm] (r2) at (6.5+1.5,0) {$v$};
\draw[->] (r2) to [bend left] node[auto]{$7$} (b2);
\draw[->] (b2) to [bend left] node[auto]{$1$} (l2);
\draw[->,very thick] (l2) to [bend left] node[auto,pos=.3]{$2$} (r2);

\node at (9.5,0) {$\xleftarrow{\mathrm{expand}_{(u,v)}}$};

\node[circle,draw,minimum size=5mm] (t3) at (11,0) {$u$};
\node[circle,draw,minimum size=5mm] (b3) at (12.5,-2) {$x$};
\draw[->] (t3) to [bend left] node[auto]{$7$} (b3);
\draw[->] (b3) to [bend left] node[auto]{$1$} (t3);

\end{tikzpicture}
    \caption{Example for the contraction of the path $\{(u,v),(v,y)\}$ in the graph $G$ resulting in
    the graph $G''$ and the
    following expansion of the tour $\{(u,x),(x,u)\}$ in $G''$.
    The final tour in $G$ includes the contracted path.
    } \label{fig_path_contract}
    \label{fig_path_expand}
\end{figure}
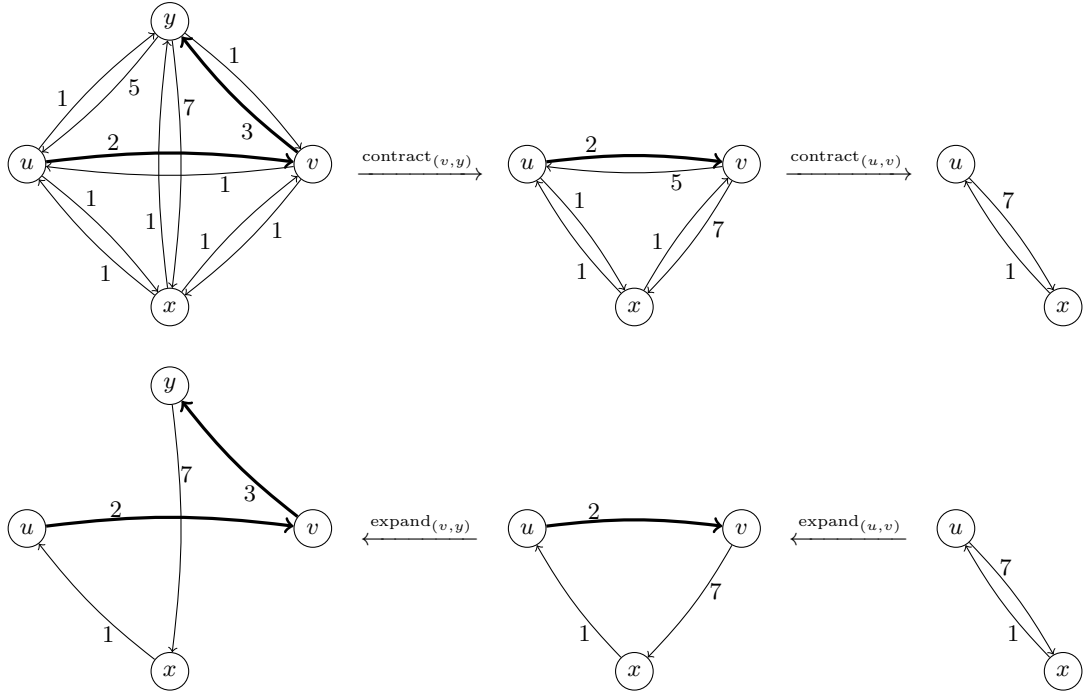

\begin{definition}
Let $G=(V,E,w)$ be some $\uint^k$-labeled complete directed graph,
let $(u,v) \in E$,
let $P \subseteq E$ be a path
$u_0,e_1,u_1,e_2,u_2,\ldots, e_r,u_r$, 
and let $Q \subseteq E$ be a set of pairwise vertex disjoint paths 
$P_1,P_2, \dots, P_r \subseteq E$.
\begin{enumerate}
\item $\mathrm{contract}_{(u,v)}(G) 
 = (V \setminus \{v\}, \{e \in E \mid \text{$v$ is not incident to $e$}\}, w')$,
where $w'(x,y)=w(x,y)$ if $x \neq u$ and $w'(u,z)=w(v,z)$.
\item $
 \mathrm{contract}_{P}(G) 
 =\mathrm{contract}_{e_1}(
    \mathrm{contract}_{e_2}(
     \dots \mathrm{contract}_{e_r}(G) \dots ))$
\item $
 \mathrm{contract}_{Q}(G) 
 =\mathrm{contract}_{P_1}(
    \mathrm{contract}_{P_2}(
     \dots \mathrm{contract}_{P_r}(G) \dots ))$
\end{enumerate}
We sometimes identify a graph with its edge set and apply $\mathrm{contract}$
directly to sets of edges and not to graphs. In this case,
we also interpret the value of $\mathrm{contract}$ as an edge set.
\end{definition}

Observe that the result of contracting several pairwise vertex disjoint paths
does not depend on the order in which the paths are contracted.

We define edge expansion in a similar manner.
Note that Definition~\ref{definition:expand} becomes essential
if $G'$ is obtained from $G$ 
by a contraction of some set $Q$ of pairwise vertex disjoint paths in $G$.

\begin{definition}\label{definition:expand}
Let $G=(V,E,w)$ and $G'=(V',E',w')$ be two $\uint^k$-labeled complete directed graphs,
let $T \subseteq E'$ be a Hamiltonian cycle of $G'$,
let $(u,v) \in E$,
let $P \subseteq E$ be a path
$u_0,e_1,u_1,e_2,u_2,\ldots ,e_r,u_r$ 
and let $Q \subseteq E$ be a set of pairwise vertex disjoint paths 
$P_1,P_2, \dots, P_r \subseteq E$.
\begin{enumerate}
\item $\mathrm{expand}_{(u,v)}(T) 
 = \{(x,y) \in T \mid x \neq u\}
 \cup \{(u,v)\}
 \cup \{(v,x) \mid (u,x) \in T\}$
\item $
 \mathrm{expand}_{P}(T) 
 =\mathrm{expand}_{e_{r}}(
    \mathrm{expand}_{e_{r-1}}(
     \dots \mathrm{expand}_{e_1}(T) \dots ))$
\item $
 \mathrm{expand}_{Q}(T) 
 =\mathrm{expand}_{P_r}(
    \mathrm{expand}_{P_{r-1}}(
     \dots \mathrm{expand}_{P_1}(T) \dots ))$
\end{enumerate}
\end{definition}

Again observe that the result of expanding several pairwise vertex disjoint paths
does not depend on the order in which the paths are expanded.

\begin{proposition}\label{prop:expanded_cycle}
Let $G=(V,E,w)$ be some $\uint^k$-labeled complete directed graph,
$Q \subseteq E$ be a set of pairwise vertex disjoint paths, and
$G'=(V',E',w')=\mathrm{contract}_Q(G)$.
For any Hamiltonian cycle $T' \subseteq E'$ of $G'$,
the edges $T=\mathrm{expand}_Q(T')$ form a Hamiltonian cycle of $G$
with $w(T) = w'(T') + w(Q)$.
\end{proposition}


\paragraph{Approximation Algorithm.} First we prove that the following algorithm computes a
$(\nicefrac{1}{2}-\eps)$-approximation for $k$-\ATSP.  Then
Theorem~\ref{thm_approx_maxtsp} shows that a modification of the algorithm
provides a $\nicefrac{1}{2}$-approximation.

\vspace{5mm}

\begin{algorithm}[H]
\algosettings
\caption{\textbf{Algorithm}: \texttt{2\ATSPApprox($V,E,w,\eps$)}}
\Input{$\uint^{2k}$-labeled complete directed graph $G=(V,E,w)$ 
 and even $\card{V}$}
\Output{set of Hamiltonian cycles of $G$}
\BlankLine
\ForEach{$F \subseteq E$ with $\card{F} \leq 2k$ that is a
    set of vertex disjoint paths}
{
 $G':= \textrm{contract}_{F}(G)$\label{alg:contract}\;
 ${\cal M} := \textrm{2\kMMApprox}(G',\eps)$\label{alg:match}\;
 \ForEach{$M \in {\cal M}$}{
  extend $M$ in an arbitrary way to a Hamiltonian cycle $T'$ in $G'$\;
  output $\textrm{expand}_{F}(T')$\;
 }
}
\end{algorithm}

\begin{lemma}\label{lem:tsp_asymp_fptas}
Let $G=(V,E,w)$ be an $\uint^{2k}$-labeled complete directed graph
with an even number of vertices, $\eps>0$, 
and $T \subseteq E$ some Hamiltonian cycle in $G$.
With probability at least $\nicefrac{1}{2}$,
\texttt{2}\ATSPApprox$(V,E,w,\eps)$ outputs a $(\nicefrac{1}{2}-\eps)$-approximation of
$T$
within time polynomial in $|(V,E,w)|+\nicefrac{1}{\eps}$.
\end{lemma}

\begin{proof}
Let $G=(V,E,w)$ be an $\uint^{2k}$-labeled complete directed graph
with even $m=\card{V}$,
and let $T$ be some arbitrary Hamiltonian cycle in $G$.

\begin{claim}\label{claim:lemmaAppliedToT}
There is a set $F$ of vertex disjoint paths in $T$ 
with $\card{F} \leq 2k$ 
such that there is a matching $M'$ 
in the graph $(V',E',w') =\textrm{contract}_{F}(G)$
with $w'(M') \geq \frac{1}{2} w(T) - w(F) $.
\end{claim}

\begin{proof}
We apply Lemma \ref{lem:combinatorial} to the sequence of edge weights of $T$.
Having an even number of edges, we can write $T$ sequentially as
\begin{align*}
T = u_1,e_1,v_1,f_1,u_2,e_2,v_2,f_2,\dots,u_p,e_p,v_p,f_p,u_1
\end{align*}
where $u_i,v_i \in V$ and $e_i,f_i \in T$.

Since $w(e_i),w(f_i) \in \uint^{2k}$,
Lemma \ref{lem:combinatorial} shows that there exist
$k' \in \{0, \ldots, k\}$ and natural numbers 
$1 \le a_1 \le b_1 < a_2 \le b_2 \cdots < a_{k'} \le b_{k'} \le p$
such that for $I = \bigcup_{i \in \{1,\ldots,k'\}} \{a_i, a_i+1, \ldots, b_i\}$,
\begin{align}
w(f_{b_1}) + w(f_{b_2}) + \cdots + w(f_{b_{k'}}) + 
\sum_{i \in I} w(e_i) + \sum_{i \notin I} w(f_i) \;\;\ge\;\; 
\frac{1}{2} \sum_{i=1}^{p}(w(e_i) + w(f_i)) \label{eqn:lemma_application_TSP}.
\end{align}

Let 
$S = \{f_{b_1}, f_{b_2}, \dots, f_{b_{k'}}\}
\cup \{e_i \mid i \in I\}
\cup \{f_i \mid i \notin I\}$.
Observe that it is possible that adjacent edges are contained in $S$.
Figure \ref{fig:matching_in_cycle} gives an example.

\begin{figure}
\centering
\footnotesize
\begin{tikzpicture}[scale=1.2]
\draw[thick](0.2,0) -- node[above]{$e_{j}$} (1,0);
\draw[thick](1,0) -- node[above]{$f_{j}$}(2,0);
\draw[dotted](2,0)  -- node[above]{$e_{j+1}$}(3,0);
\draw[thick](3,0) -- node[above]{$f_{j+1}$}(4,0);
\draw[thick](4,0)  -- node[above]{$e_{j+2}$}(5,0);
\draw[dotted](5,0) -- node[above]{$f_{j+2}$}(6,0);
\draw[thick](6,0)  -- node[above]{$e_{j+3}$}(7,0);
\draw[thick](7,0) -- node[above]{$f_{j+3}$}(8,0);
\draw[thick](8,0)  -- node[above]{$e_{j+4}$}(9,0);
\draw[thick](9,0) -- node[above]{$f_{j+4}$}(10,0);
\draw[dotted](10,0)  -- node[above]{$e_{j+5}$}(11,0);
\draw[thick](11,0) -- node[above]{$f_{j+5}$}(11.8,0);
\draw[](12,0) node[right]{$\dots$};
\draw[](0,0) node[left]{$\dots$};
\foreach \x in {1,2,3,4,5,6,7,8,9,10,11}
\draw[fill=black] (\x,0) circle (1pt);
\draw[rounded corners=10pt] 
	(6,-1) -- 
	(8.1,-1) --
	(8.1,1) --
	(3.9,1) --
	(3.9,-1) --
	(6,-1);
\draw[rounded corners=10pt] 
	(0.2,-1) -- 
	(2.1,-1) --
	(2.1,1) --
	(0.2,1);
\draw[rounded corners=10pt] 
	(9,-1) -- 
	(10.1,-1) --
	(10.1,1) --
	(7.9,1) --
	(7.9,-1) --
	(9,-1);
\draw[](2,1.2) node[left]{$b_i$};
\draw[](4,1.2) node[right]{$a_{i+1}$};
\draw[](8,1.2) node[left]{$b_{i+1}$};
\draw[](8,1.2) node[right]{$a_{i+2}$};
\draw[](10,1.2) node[left]{$b_{i+2}$};
\end{tikzpicture}
\caption{\footnotesize Some part of the cycle $T$, 
where $S \subseteq T$ contains the depicted edges
and is partially defined by 
$b_i=j$, $a_{i+1}=j+2$, $b_{i+1}=j+3$, and $a_{i+2}=b_{i+2}=j+4$.}
\label{fig:matching_in_cycle}
\end{figure}
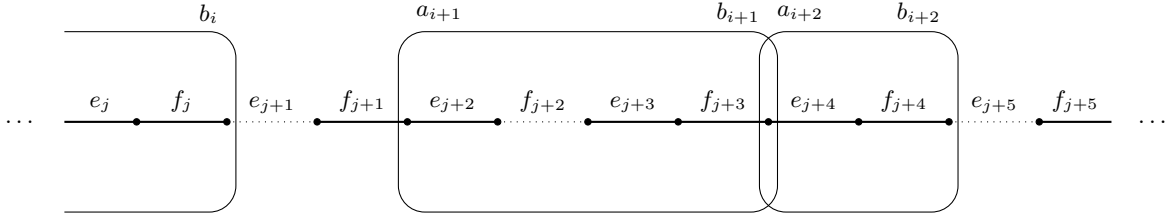

Observe that for $1 \le j \le p$ and $f_{0}=f_p$ the following holds.
\begin{eqnarray}
    f_{j-1},e_{j} \in S &\;\;\iff\;\;& \exists i \in \{1,\ldots,k'\}, j=a_i \label{eqn_cases_a}\\
    e_j,f_j \in S &\;\;\iff\;\;& \exists i \in \{1,\ldots,k'\}, j=b_i \label{eqn_cases_b}
\end{eqnarray}

Let $F=\{e_{a_1},f_{b_1},e_{a_2},f_{b_2},\dots,e_{a_{k'}},f_{b_{k'}}\}$
and note that $\card{F}=2k'$.
We argue that contracting $F$ 
will transform any path in $S$ into a single edge
such that the resulting edge set is a matching:

Suppose $S$ contains some path $P=\{e_r,f_r,e_{r+1},f_{r+1},\dots,e_s,f_s\}$, where
we assume $P$ to be maximal (i.e., $f_{r-1},e_{s+1} \notin S$).
From (\ref{eqn_cases_a}), (\ref{eqn_cases_b}), and $a_1 \le b_1 < a_2 \le b_2 \cdots < a_{k'} \le b_{k'}$
we can draw the following conclusions:
\begin{itemize}
    \item $e_r,f_r \in P$ yields $r=b_i$ for some $1 \leq i \leq k'$
    \item $f_r,e_{r+1},f_{r+1} \in P$ yields $r+1=a_{i+1}=b_{i+1}$
    \item $f_{r+1},e_{r+2},f_{r+2} \in P$ yields $r+2=a_{i+2}=b_{i+2}$\\[1.5mm]
    \hspace*{35mm} $\vdots$
    \item $f_{s-1},e_s,f_s \in P$ yields $s=a_{s-r+1}=b_{s-r+1}$
\end{itemize}

Note that $e_r \notin F$, since otherwise $e_{b_i} \in F$,
hence $a_i=r$, and by (\ref{eqn_cases_a}), $f_{r-1} \in S$, which contradicts the maximality of $P$.
Therefore, contracting $\{f_{b_i},e_{a_{i+1}},f_{b_{i+1}},\dots,f_{b_j}\} \subseteq F$
transforms $P$ into the single edge $e_r$.
A similar argumentation shows the same result
for paths that start with some edge $f_r$ or end with some edge $e_s$.
Hence, contracting $F$
transforms every path in $S$ into a single edge,
and $M' =\textrm{contract}_{F}(S)$
is a matching in the graph $(V',E',w')=\textrm{contract}_{F}(G)$.
We further obtain
\begin{align*}
w'(M')
&= w(S) - w(F)\\
&= w(f_{b_1}) + w(f_{b_2}) + \cdots + w(f_{b_{k'}}) 
+ \sum_{i \in I} w(e_i) + \sum_{i \notin I} w(f_i) - w(F)\\
&\stackrel{(\ref{eqn:lemma_application_TSP})}{\geq}
\frac{1}{2} \sum_{i=1}^{p}(w(e_i) + w(f_i)) - w(F)\\
&= \frac{1}{2} w(T) - w(F)
\end{align*}
which proves the claim.
\end{proof}

We fix the iteration of \texttt{2}\ATSPApprox\
where the algorithm chooses $F$ as in the claim.
By Claim \ref{claim:lemmaAppliedToT}
we know that there is a matching $M'$ of $G'=(V',E',w')$ with 
\begin{align}
w'(M') \geq \frac{1}{2} w(T) - w(F)\label{eqn:matching_in_g'}.
\end{align}
Hence, with probability at least $\nicefrac{1}{2}$
the set ${\cal M}$ contains some matching $M$ of $G'$
such that
\begin{align*}
w'(M)
&\geq (1-\eps)w'(M')\\
&\stackrel{(\ref{eqn:matching_in_g'})}{\geq} (1-\eps)\left( \frac{1}{2} w(T) -
w(F)\right).
\end{align*}
We extend $M$ to some Hamiltonian cycle $T'$ of $G'$ in an arbitrary way
without losing weights.
By Proposition~\ref{prop:expanded_cycle}
we can expand $T'$ with $F$ and obtain
a Hamiltonian cycle $\tilde{T}$ in $G$ with
\begin{align*}
w(\tilde{T}) 
&= w'(T') + w(F)\\
&\geq w'(M) + w(F)\\
&\geq (1-\eps)\left(\frac{1}{2} w(T) - w(F)\right) + w(F)\\
&\geq (1-\eps) \frac{1}{2} w(T) + \eps w(F)\\
&\geq (1-\eps) \frac{1}{2} w(T).
\end{align*}
Moreover, the running time of every operation of the algorithm,
including the execution of the randomized maximum matching algorithm,
and the number of iterations of the loops are
polynomial in the length of the input and $\nicefrac{1}{\eps}$,
which completes the proof of the lemma.
\end{proof}

\begin{theorem} \label{thm_approx_maxtsp}
$k$-\ATSP\ is randomized $\nicefrac{1}{2}$-approximable.
\end{theorem}

\begin{proof}
Let $G=(V,E,w)$ be an $\uint^{2k}$-labeled complete directed graph
with even $m=\card{V}$,
and let $T$ be some arbitrary Hamiltonian cycle in $G$.
The proof can easily be extended to graphs
with an odd number of vertices or objectives.

For each $1 \leq i \leq 2k$ we choose an $f_i \in T$ 
with $w(f_i) \geq \frac{1}{m} w(T)$
(i.e., $f_i$ is a heaviest edge of $T$ with respect to component $i$).
We let $F \subseteq T$ be a smallest set with even cardinality containing all the
$f_i$.
We get $\card{F}\leq 2k$ and
\begin{align} w(F) \geq \frac{1}{m} w(T). \label{eqn:0} \end{align}
$F$ is a set of vertex disjoint paths in $T$
and hence can be used to contract edges in $G$ and $T$.
Let $G' = (V',E',w') = \textrm{contract}_{F}(G)$ and $T' = \textrm{contract}_{F}(T)$.
Clearly, $T'$ is a Hamiltonian cycle in $G'$.
Moreover, $G'$ has an even number of vertices.
By Lemma~\ref{lem:tsp_asymp_fptas},
we can find in polynomial time a Hamiltonian cycle $\tilde{T}$ in $G'$
such that $w'(\tilde{T}) \geq (1-\eps)\frac{1}{2}w'(T')$, where $\eps=\nicefrac{1}{m}$.
Moreover, we can expand $\tilde{T}$ with $F$
to obtain a Hamiltonian cycle $\hat{T}$ in $G$ such that
\begin{align*}
w(\hat{T})
&= w'(\tilde{T}) + w(F)\\
&\geq (1-\eps) \frac{1}{2} w'(T') + w(F)\\
&= \left(\frac{1}{2} w'(T') +  \frac{1}{2} w(F)\right) + 
 \frac{1}{2} w(F) - \frac{\eps}{2} w'(T') \\
&= \frac{1}{2} w(T) + \frac{1}{2} (w(F) - \eps w'(T'))\\
&\stackrel{(*)}{\geq} \frac{1}{2} w(T),
\end{align*}
where ($*$) follows from
\begin{align*}
w(F) - \eps w'(T')
&= w(F) - \frac{1}{m} w'(T')\\
&\geq w(F) - \frac{1}{m} w(T)\\
&\stackrel{\eqref{eqn:0}}{\geq} \frac{1}{m} w(T) - \frac{1}{m} w(T)\\
&= 0.
\end{align*}

Note that, although we do not know the set $F$ of heaviest edges in the Hamiltonian cycle,
we can simply try all possible sets of heaviest edges,
since the number of objectives $2k$ is constant.
\end{proof}


\subsection{\texorpdfstring{$\boldsymbol{k}$}{k}-Objective Maximum Weighted Satisfiability} \label{sec_sat}

\paragraph{Definitions.}
We consider formulas over a finite set of propositional variables $V$.
A \emph{literal} is a propositional variable $v\in V$ 
or its negation $\overline{v}$,
a \emph{clause} is a finite, nonempty set of literals,
and a \emph{formula in conjunctive normal form} (\textbf{CNF}, for short)
is a finite set of clauses.
A \emph{truth assignment} is a mapping $I \colon V \to \{0,1\}$.
For some $v \in V$,
we say that $I$ 
\emph{satisfies the literal $v$} if $I(v)=1$,
and $I$ \emph{satisfies the literal $\overline{v}$} if $I(v)=0$.
We further say that $I$ \emph{satisfies the clause $C$} and write $I(C)=1$ 
if there is some literal $l \in C$ that is satisfied by $I$,
and $I$ \emph{satisfies a formula in CNF} if $I$ satisfies all of its clauses.
For a set of clauses $\hat{H}$ and a variable $v$ let $\hat{H}[v] =
\{C \in \hat{H} \mid v \in C\}$ be the set of clauses 
that are satisfied
if this variable
is set to one, and analogously $\hat{H}[\oli{v}] =
\{C \in \hat{H} \mid \oli{v} \in C\}$ be the set of clauses
that are satisfied
if this variable
is set to zero.

Given a formula in CNF and a $k$-objective weight function that maps each clause
to a $k$-objective weight, our goal is to find truth assignments
that maximize the sum of the weights of all satisfied clauses.

\begin{quote}
    \begin{tabbing}
        {\bf $\boldsymbol{k}$-Objective Maximum Weighted Satisfiability
        (\kMaxSATbf)}\\
        Instance: \= Formula $H$ in CNF over a set of variables $V$,
        weight function $w \colon H \to \uint^k$\\
        Solution: \> Truth assignment $I \colon V \to \{0,1\}$\\
        Weight: \> Sum of the weights of all clauses satisfied by $I$, i.\,e.,
         $w(I) = \Sum_{\substack{C \in H\\I(C) = 1}} w(C)$
    \end{tabbing}
\end{quote}

\paragraph{High-Level Explanation of the Algorithm.}
We apply the balancing results to \kMaxSAT.
For a given formula $H$ in CNF over the variables $V$, the strategy is as follows:
Start with a list of the variables $V$ and guess a partition of this list into
$2k$ consecutive intervals.
Assign $1$ to the variables in every second interval and $0$ to the remaining variables.
The balancing results assure the existence of a partition that yields an assignment
whose weights are approximately one half of the \emph{total} weights of $H$,
up to an error induced by the variables at the boundaries of the partition.
The error can be removed by first guessing a satisfying assignment for several
influential variables $V^0$ of the formula.
This results in a $\nicefrac{1}{2}$-approximation for \kMaxSAT.

\paragraph{Approximation Algorithm.}
We show that the following algorithm 
computes a $\nicefrac{1}{2}$-approx\-ima\-tion for \kMaxSAT[2k].

\vspace{5mm}

\begin{algorithm}[H]
    \algosettings
    \caption{\textbf{Algorithm}: \texttt{\kWSATApprox($H,w$)}}
    \Input{Formula $H$ in CNF over the variables $V=\{v_1,\dots,v_m\}$,
     $2k$-objective weight function $w \colon H \to \uint^{2k}$}
    \Output{Set of truth assignments $I \colon V \to \{0,1\}$}
    \BlankLine

    \ForEach{$V^0 \subseteq V$ with $\card{V^0} \le (2k)^2$}{
        let $I(v) := 0$ for all $v \in V^0$\;
        $G := \{C \in H \mid \neg \exists v \in V^0\,\,(\oli{v} \in C)\}$\;
        $V^1 := \{v \in V \setminus V^0 \mid 
        2k \cdot w(G[\oli{v}])\not\leq w(H \setminus G)\}$\;
        set $I(v) := 1$ for all $v \in V^1$\;
        $V' := V \setminus (V^0 \cup V^1)$\;
        \ForEach{$a_1,b_1,a_2,b_2,\dots,a_k,b_k \in \{i \mid v_i \in V'\}$}{
         \ForEach{$v_i \in V'$}{
          \lIf{$\exists j (a_j \le i \le b_j)$}{$I(v_i):=1$}
          \lElse{$I(v_i):=0$}
         }
         output $I$
        }
    }
\end{algorithm}

\begin{theorem}
    \kMaxSAT\ is $\nicefrac{1}{2}$-approximable.
\end{theorem}
\begin{proof}
In the following, we assume without loss of generality that the number of objectives $2k$ is even.
We show that the approximation is realized by the algorithm \kWSATApprox.
First note that this algorithm runs in polynomial time since $k$ is constant.
For the correctness, let $(H,w)$ be the input where $H$ is a formula over the
variables $V=\{v_1,\dots,v_m\}$ and $w \colon H \to \uint^{2k}$ is the
$2k$-objective weight function.
Let $I_o\colon V \rightarrow \{0,1\}$ be an optimal truth assignment.
We show that there is an iteration of the loops of \kWSATApprox($H,w$)
that outputs a truth assignment $I$ such that $w(I) \ge w(I_o)/2$.
First we show that there is an iteration of the first loop that uses a suitable
set $V^0$.

\begin{claim}\label{claimv0}
There is some set $V^0 \subseteq \{v \in V\mid I_o(v) = 0\}$ with $\card{V^0}
\le (2k)^2$
such that for
$G = \{C \in H \mid \neg \exists v \in V^0\,\,(\oli{v} \in C)\}$ and any
$v \in V \setminus V^0$ it holds that
\begin{align*}
2k \cdot w(G[\oli{v}]) \not\leq w(H \setminus G) \qquad \Longrightarrow \qquad
I_o(v)=1.
\end{align*}
\end{claim}
\begin{proof}
As a special case, if $\card\{v \in V
\mid I_o(v) = 0\} < (2k)^2$, the assertion obviously holds for $V^0 = \{v \in V \mid I_o(v)
= 0\}$, since
$I_o(v)=1$ for all $v \in V\setminus V^0$.

Otherwise, let $V^0 =
\{u_{2kt+r} \mid r=1,2,\dots,2k$ and $ t=0,1,\dots,2k-1\}$,
where the $u_{2kt+r}\in V$ are
defined inductively in the following way:
\begin{itemize}[noitemsep]
\item[(IB)] $H_0 := H$
\item[(IS)] $2kt+r-1 \to 2kt+r$:
\begin{itemize}[noitemsep]
\item choose $v \in V \setminus \{u_1,\dots,u_{2kt+r-1}\}$ such that $I_o(v) = 0$ and $w_r(H_{2kt+r-1}[\oli{v}])$ is maximal
\item $u_{2kt+r} := v$
\item $H_{2kt+r} := H_{2kt+r-1} \setminus H_{2kt+r-1}[\oli{v}]$
\item $\alpha_{2kt+r} := w(H_{2kt+r-1}[\oli{v}])$.
\end{itemize}
\end{itemize}

We now show that the stated implication holds, so let
$v \in V \setminus V^0$ and $j \in \{1,2,\dots,2k\}$ such that
$2k \cdot w_j(G[\oli{v}]) > w_j(H \setminus G)$. Because the union $\bigcup_{i=1}^{4k^2}
H_{i-1}[\oli{u_i}] = H \setminus G$ is disjoint, we get
\begin{align*}
w(H \setminus G) = \sum_{r=1}^{2k}\sum_{t=0}^{2k-1} \alpha_{2kt+r} \ge
\sum_{t=0}^{2k-1} \alpha_{2kt+j}
\intertext{and thus}
w_j(G[\oli{v}]) > \sum_{t=0}^{2k-1} \frac{(\alpha_{2kt+j})_j}{2k}.
\end{align*}
Hence, by a pigeonhole argument, there must be some $t \in \{0,1,\dots,2k-1\}$ such that
$w_j(G[\oli{v}]) > (\alpha_{2kt+j})_j$. But since $G \subseteq H_{2kt+j-1}$
and thus even $w_j(H_{2kt+j-1}[\oli{v}]) > w_j(G[\oli{v}]) >
(\alpha_{2kt+j})_j$, the only reason we did not choose $v$ in the iteration
$2kt+j$ (or even earlier) is that $I_o(v) = 1$.
\end{proof}

We choose the iteration of the algorithm where $V^0$ equals the set
whose existence is guaranteed
by Claim~\ref{claimv0}. Furthermore let $G$ and $V^1$ be defined as in the
algorithm
and observe that by the claim it holds that $I_o(v)=1$ for all $v \in V^1$.
Since $I_o(v) = 0$ for all $v \in V^0$, the truth assignment $I$ defined
in the algorithm coincides with $I_o$ on $V^0 \cup V^1$.

Let further $V' = V \setminus (V^0\cup V^1)$ and 
$G' = \{ C \in G \mid \neg \exists v \in V^0\,\,(\oli{v} \in C) \land \neg
\exists v \in V^1\,\,(v \in C) \land \exists v \in V'\,\, (v \in C \lor
\oli{v} \in C)\}$ be the set of clauses that are
not yet satisfied by $I$ but that could be satisfied by further extending
$I$.

Now we apply the balancing result. Let $L' = V' \cup \{\oli{v} \mid v \in V'\}$.
For $v_i \in V'$ let
\begin{align*}
x_i&=\sum_{C \in G'[v_i]} \frac{w(C)}{\card{(C \cap L')}}
&\mbox{and}&
&y_i&=\sum_{C \in G'[\oli{v_i}]} \frac{w(C)}{\card{(C \cap L')}},
\end{align*}
and for $v_i \in V^0\cup V^1$ let
$$x_i = y_i = 0.$$
It holds that
\begin{align*}
\sum_{v_i \in V} x_i + y_i = \sum_{v_i \in V'} x_i + y_i = w(G').
\end{align*}
Note that for all $v_i \in V'$, we have the bound
$y_i \le w(G'[\oli{v_i}]) \le w(G[\oli{v_i}]) \le \frac{1}{2k} w(H \setminus G)$
because of the definition of $V'$ and $V^1$.
Hence, for all $v_i \in V$,
$$y_i \le \frac{1}{2k} w(H \setminus G).$$
If we scale all values $x_i$ and $y_i$ to natural numbers,
then by Corollary~\ref{coro:combinatorial}, there exist $k' \le k$ disjoint,
nonempty intervals $J_1, \ldots, J_{k'} \subseteq \{1, \ldots, m\}$ such
that for $J = J_1 \cup \dots \cup J_{k'}$ it holds that
\begin{align*}
\sum_{i \in J} x_i + \sum_{i \notin J} y_i \ge
\frac{1}{2} w(G') - k' \frac{1}{2k} w(H \setminus G) \ge \frac{1}{2}(w(G') -
w(H \setminus G)).
\end{align*}
The algorithm tries all combinations of $k$ (possibly empty) intervals
$J_1 = [a_1,b_1], \ldots, J_k = [a_k,b_k]$.
In particular, it will test the combination of the $k'$
nonempty intervals mentioned in Corollary~\ref{coro:combinatorial}.
For $I$ being the truth assignment generated in this iteration it holds that
\begin{align}\label{eqn227736}
w(\{C \in G' \mid I(C) = 1\}) \ge
\sum_{i \in J} x_i + \sum_{i \notin J} y_i \ge \frac{1}{2}(w(G') -w(H \setminus G)).
\end{align}
Furthermore, since $I$ and $I_o$ coincide on $V \setminus V'$, we have
\begin{align}
w(\{C \in H \setminus G' \mid I(C) = 1\}) &=
w(\{C \in H \setminus G' \mid I_o(C) = 1\})\label{eqn2343}\\
&\ge w(\{C \in H \setminus G \mid I_o(C) = 1\})\notag\\
& = w(\{H \setminus G\}).\label{eqn33438}
\end{align}
Thus we finally obtain
\begin{align*}
w(I) &\;=\; w(\{C \in H \setminus G' \mid I(C) = 1\}) + w(\{C \in G'\mid I(C) = 1\})\\
&\stackrel{\eqref{eqn227736}}{\ge} w(\{C \in H \setminus G' \mid I(C) = 1\}) +
\tfrac{1}{2}(w(G') - w(H \setminus G))\\
&\stackrel{\eqref{eqn2343}}{=} w(\{C \in H \setminus G' \mid I_o(C) = 1\}) +
\tfrac{1}{2}(w(G') - w(H \setminus G))\\
&\stackrel{\eqref{eqn33438}}{\ge} \tfrac{1}{2} w(\{C \in H \setminus G' \mid I_o(C) = 1\}) +
\tfrac{1}{2}w(G')\\
&\; \ge\; \tfrac{1}{2}w(I_o).\qedhere
\end{align*}
\end{proof}


\bibliographystyle{alpha}

\begin{thebibliography}{KLSS05}

\bibitem[AW02]{AW02}
T.~Asano and D.~P. Williamson.
\newblock Improved approximation algorithms for {MAX} {SAT}.
\newblock {\em Journal of Algorithms}, 42(1):173--202, 2002.

\bibitem[BMP08]{BMP08}
M.~Bl{\"a}ser, B.~Manthey, and O.~Putz.
\newblock Approximating multi-criteria {M}ax-{TSP}.
\newblock In {\em ESA}, pages 185--197, 2008.

\bibitem[EG00]{eg00}
M.~Ehrgott and X.~Gandibleux.
\newblock A survey and annotated bibliography of multiobjective combinatorial
  optimization.
\newblock {\em OR Spectrum}, 22(4):425–460, 2000.

\bibitem[Ehr05]{ehr05}
M.~Ehrgott.
\newblock {\em Multicriteria Optimization}.
\newblock Springer Verlag, 2005.

\bibitem[EK01]{EK01}
L.~Engebretsen and M.~Karpinski.
\newblock Approximation hardness of {TSP} with bounded metrics.
\newblock In {\em ICALP '01: Proceedings of the 28th International Colloquium
  on Automata, Languages and Programming}, pages 201--212, London, UK, 2001.
  Springer-Verlag.

\bibitem[FNW79]{FNW79}
M.~L. Fisher, G.~L. Nemhauser, and L.~A. Wolsey.
\newblock An analysis of approximations for finding a maximum weight
  {H}amiltonian circuit.
\newblock {\em Operations Research}, 27(4):799--809, 1979.

\bibitem[GW94]{GW94}
M.~X. Goemans and D.~P. Williamson.
\newblock New 3/4-approximation algorithms for the maximum satisfiability
  problem.
\newblock {\em SIAM Journal on Discrete Mathematics}, 7(4):656--666, 1994.

\bibitem[GW95]{GW95}
M.~X. Goemans and D.~P. Williamson.
\newblock Improved approximation algorithms for maximum cut and satisfiability
  problems using semidefinite programming.
\newblock {\em Journal of the ACM}, 42(6):1115--1145, 1995.

\bibitem[H{\aa}s97]{Has97}
J.~H{\aa}stad.
\newblock Some optimal inapproximability results.
\newblock In {\em STOC}, pages 1--10, 1997.

\bibitem[Joh74]{Joh74}
D.~S. Johnson.
\newblock Approximation algorithms for combinatorial problems.
\newblock {\em Journal of Computer System Sciences}, 9(3):256--278, 1974.

\bibitem[KLSS05]{KLS+05}
H.~Kaplan, M.~Lewenstein, N.~Shafrir, and M.~Sviridenko.
\newblock Approximation algorithms for asymmetric {TSP} by decomposing directed
  regular multigraphs.
\newblock {\em Journal of the ACM}, 52(4):602--626, 2005.

\bibitem[Llo78]{ll78-degree-theory}
N.~G. Lloyd.
\newblock {\em Degree Theory}.
\newblock Cambridge University Press, Cambridge, England, 1978.

\bibitem[Man09]{man09}
B.~Manthey.
\newblock On approximating multi-criteria {TSP}.
\newblock In {\em Proceedings of 26th Annual Symposium on Theoretical Aspects
  of Computer Science}, volume 09001 of {\em Dagstuhl Seminar Proceedings},
  pages 637--648. Internationales Begegnungs- und Forschungszentrum fuer
  Informatik (IBFI), Schloss Dagstuhl, Germany, 2009.

\bibitem[PY91]{PY91}
C.~H. Papadimitriou and M.~Yannakakis.
\newblock Optimization, approximation, and complexity classes.
\newblock {\em Journal of Computer System Sciences}, 43(3), 1991.

\bibitem[PY00]{PY00}
C.~H. Papadimitriou and M.~Yannakakis.
\newblock On the approximability of trade-offs and optimal access of web
  sources.
\newblock In {\em FOCS '00: Proceedings of the 41st Annual Symposium on
  Foundations of Computer Science}, pages 86--95, Washington, DC, USA, 2000.
  IEEE Computer Society.

\bibitem[SBLL09]{SBLL09}
R.~Santana, C.~Bielza, J.~A. Lozano, and P.~Larra{\~{n}}aga.
\newblock Mining probabilistic models learned by {EDA}s in the optimization of
  multi-objective problems.
\newblock In {\em GECCO '09: Proceedings of the 11th Annual conference on
  Genetic and evolutionary computation}, pages 445--452, New York, NY, USA,
  2009. ACM.

\bibitem[Yan94]{Yan94}
M.~Yannakakis.
\newblock On the approximation of maximum satisfiability.
\newblock {\em Journal of Algorithms}, 17(3):475--502, 1994.

\end{thebibliography}

\end{document}